\documentclass[DIV=classic,a4paper,10pt]{myart}

\KOMAoptions{DIV=last}

\usepackage{caption}
\usepackage{subcaption}
\usepackage{graphicx}
\usepackage{float}

\usepackage[normalem]{ulem}

\usepackage{accents}
\newcommand{\ubar}[1]{\underaccent{\bar}{#1}}

\begin{document}

\title{Computational aspects of robust optimized certainty equivalents and option pricing}
\tnotetext[t]{We thank Martin Herdegen, Jan Ob{\l}{\'o}j, Mathias Pohl, and Johannes Wiesel
for fruitful discussions}

\author[a,1,w]{Daniel Bartl}
\author[b,2,u]{Samuel Drapeau}
\author[c,3,v]{Ludovic Tangpi}

\address[a]{Department of Mathematics, University of Konstanz, Universit{\"a}tsstra{\ss}e 10, 78464 Konstanz- Germany.}
\address[b]{School of Mathematical Sciences and Shanghai Advanced Institute of Finance (SAIF/CAFR), Shanghai Jiao Tong University, 
211 West Huaihai road, Shanghai, P.R. 200030 China.}
\address[c]{Department of Operations Research and Financial Engineering, Princeton University, Charlton Street Princeton, NJ 08544 USA}

\eMail[1]{daniel.bartl@uni-konstanz.de}
\eMail[2]{sdrapeau@saif.sjtu.edu.cn}
\eMail[3]{ndounkeu@princeton.edu}

\myThanks[w]{Vienna Science and Technology Fund (WWTF) through project VRG17-005 and FWF grant Y00782 are gratefully acknowledged.}
\myThanks[u]{Financial support from the National Science Foundation of China, grant 11550110184 and 11671257 are gratefully acknowledged.
Grant "Assessment of Risk and Uncertainty in Finance" number AF0710020 from Shanghai Jiao Tong University is gratefully acknowledged.}
\myThanks[v]{Financial support from Vienna Science and Technology Fund (WWTF) under Grant MA 14-008 is gratefully acknowledged.}

\abstract{
    Accounting for model uncertainty in risk management and option pricing leads to infinite dimensional optimization problems which are both analytically and numerically intractable.
    In this article we study when this hurdle can be overcome for the so-called optimized certainty equivalent risk measure (OCE) -- including the average value-at-risk as a special case.
    First we focus on the case where the uncertainty is modeled by a nonlinear expectation which penalizes distributions that are ``far'' in terms of optimal-transport distance (Wasserstein distance for instance) from a given baseline distribution.
    It turns out that the computation of the robust OCE reduces to a finite dimensional problem, which in some cases can even be solved explicitly.
    This principle also applies to the shortfall risk measure as well as for the pricing of European options.
    Further, we derive convex dual representations of the robust OCE for measurable claims without any assumptions on the set of distributions.
    Finally, we give conditions on the latter set under which the robust average value-at-risk is a tail risk measure.
}

\date{\today}
\keyAMSClassification{91G80, 90B50, 60E10, 91B30}
\keyWords{Optimized certainty equivalent; optimal transport; Wasserstein distance; penalization, distribution uncertainty; convex duality; average value-at-risk; robust option pricing.}
\ArXiV{1706.10186}
\maketitle

\section{Introduction}
In this article we study properties of the optimized certainty equivalent (OCE for short) of Ben-Tal and Teboulle \cite{ben-tal02,Ben-Teb}, and to a wider extent option pricing, under model uncertainty.
In the context of risk assessment, the rationale behind the definition of the OCE is as follows.
Assume that a financial agent faces a future uncertain loss profile with distribution $\mu$, and she wants to assess its risk.
In her assessment, given a loss function $l\colon\mathbb{R}\to (-\infty,\infty]$, she computes the expectation $\int_{}^{} l(x)\mu(dx)$ representing the present average cost of her losses.
She can however reduce her overall future losses by allocating some cash $m$, resulting in the present value $\int_{}^{}l(x-m)\mu(dx)+m$.
Minimizing over all possible allocations defines the optimal cost or OCE of $\mu$ with respect to the loss function $l$:
\begin{equation}
    \label{eq:oce_classical}
    \mathrm{OCE}(l):=\inf_{m \in \mathbb{R}}\left( \int l(x-m)\mu(dx)+m \right).
\end{equation}
Now, if the distribution $\mu$ of future loss is not perfectly known, the risk averse financial agent will consider the overall cost of an allocation $m$ to be given by $\mathcal{R}(l(\cdot-m))+m$, where $\mathcal{R}$ is a nonlinear expectation.
To wit, $\mathcal{R}$ models the degree of conservatism or risk aversion of the investor and can be interpolated between the linear case $\mathcal{R}(\cdot) = \int \cdot\,d\mu$ and the worst case $\mathcal{R}(\cdot) = \sup_{\mu \in {\cal M}_1(\mathbb{R})}\int \cdot \,d\mu$, with ${\cal M}_1(\mathbb{R})$ the set of probability measures on $\mathbb{R}$.
Hence, the natural definition of the robust optimized certainty equivalent is
\begin{equation}
    \label{eq:roboce intro}
    \mathcal{OCE}(l) :=\inf_{m\in\mathbb{R}} \left( \mathcal{R}( l(\cdot-m))+m  \right),
\end{equation}
that is, the minimal allocation cost when the future expected loss is written in terms of the nonlinear expectation $\mathcal{R}$.

The classical OCE satisfies sound economical properties discussed in \cite{Ben-Teb}.
In particular, it is a convex monetary risk measure in the sense of \citet{artzner01} and \citet{foellmer02} which is additionally law-invariant, see \citet{fritelli03} for definition and consequences.
Furthermore, depending on the specification of the loss function $l$, it includes classical risk measures such as the entropic risk measure, the average value-at-risk, see \citet{Roc-Ury}, the monotone mean variance of \citet{MMR2006} and as a scaling limit, the shortfall risk measure of \citet{foellmer02}.
Stated as a classical unconstrained one dimensional optimization problem, the OCE is a smooth quantification instrument, see \citet{cheridito02} and \citet{kupper01}.
The computation of the risk as well as the risk contributions can be explicitly stated in terms of first order conditions and efficiently implemented using Fourier transform methods, see \citet{Dra-Kup-Pap}, and stochastic root finding methods, see \citet{Ham-Sal-Web14} as well as \citet{Dun-Web10}.
However when facing model uncertainty, these properties become a-priori challenging due to the potential infinite dimensional nature of the optimization problem \eqref{eq:roboce intro}.
In addition, it is not clear by how much the resulting robust quantification of risk deviates from its non robust counterpart, a crucial question in practice.

The goal of this paper is to study the robust OCE and provide several ways to reduce the complexity stemming from the robustness to get explicit formulas allowing for a quantification of the risk under model misspecification.
Our first main result focuses on the case where 
\begin{equation*}
    {\cal R}(f) := \sup_{\mu \in {\cal M}_1(\mathbb{R})}\left( \int f\,d\mu - \varphi\left(d_c\left(\mu_0,\mu\right)\right)\right),
\end{equation*}
with $\varphi$ being a positive function, $d_c$ an optimal transport-like distance with cost function $c(x,y)$ such as the Wasserstein distance and $\mu_0$ a fixed, given distribution.
The underlying intuition is that, using past information for instance, one knows a priori that a distribution $\mu_0$ is \emph{likely} to be the true distribution of the financial loss whose risk is being assessed.
Due to uncertainty about this estimation, one considers every possible distribution, penalizing however those that are \emph{``far away''} from the baseline distribution $\mu_0$ in terms of the distance $d_c$.
See further concrete financial motivations for such an approach in Remark \ref{rem:c.bounded.or.different.metrik}.
After proving the general duality formula
\begin{align} 
    \label{eq:rep.ball_intro}
    \sup_{\mu \in\mathcal{M}_1(\mathbb{R})}  \left( \int\nolimits f\,d\mu - \varphi\left(d_c\right(\mu_0,\mu\left)\right)\right) =\inf_{\lambda\geq 0} \left( \int\nolimits f^{\lambda c}\,d\mu_0 +\varphi^\ast(\lambda) \right)
\end{align}
with $\varphi^*$ being the convex conjugate of $\varphi$
in Theorem \ref{thm:integral.neigborhood},  we show in
Theorem \ref{thm:oce.and.es.robust.compute} that
\begin{equation*}
    \mathcal{OCE}(l)=\inf_{\lambda\geq 0} \left( \mathrm{OCE}\left(l^{\lambda c}\right) + \varphi^*(\lambda)\right),  \quad\text{where}\quad l^{\lambda c}(x):=\sup_{y\in\mathbb{R}}\left( l(y)-\lambda c(x,y)\right).
\end{equation*}
This formula for the robust OCE shows that the \emph{infinite dimensional} optimization problem of computing $\mathcal{OCE}(l)$ simplifies into a finite dimensional problem of computing the OCE of the distribution $\mu_0$ and with modified loss function $l^{\lambda c}$.

We stress that the duality formula \eqref{eq:rep.ball_intro} is interesting on its own, and is valid for measurable functions $f$, lower semicontinuous cost and penalty functions.
Adjusting the penalty function $\varphi$ enables to set the level of risk aversion.
A particular case considered in the literature arises when choosing $\varphi := \infty1_{(\delta,\infty]}$ for some $\delta>0$, as the nonlinear expectation $\mathcal{R}$ becomes
\begin{equation*}
    \mathcal{R}(f) = \sup_{\left\{\mu\in {\cal M}_1(\mathbb{R})\colon d_c(\mu_0,\mu)\le \delta\right\}}\int f\,d\mu,
\end{equation*}
the worst case expectation over a ball around the baseline distribution $\mu_0$.
In this case, taking $d_c$ to be the first order Wasserstein distance and $l(x) = x^+/\alpha$ so that the OCE becomes the average value-at-risk at level $\alpha$, one has
\begin{equation}
    \label{eq:rob avar intro}
    \mathcal{AV@R}_\alpha =\mathrm{AV@R}_\alpha(\mu_0) + \delta/\alpha,
\end{equation}
see Example \ref{ex:avar.formula} for different penalizations $\varphi$ and different distances.
In other terms, the robust average value-at-risk is the same as the standard average value-at-risk plus an ``uncertainty premium'' $\delta/\alpha$.
The above formula was obtained by \citet{Wozabal14} in the case where the existence of a dominating probability measure is assumed.
In the particular case where $\varphi = \infty1_{(\delta,\infty]}$, convex dual representations of $\mathcal{R}$ and applications thereof to stochastic programming have been recently studied.
The first contribution is due to \citet{Esfahani} who derive the duality formula when $d_c$ is the first order Wasserstein distance, $\mu_0$ the empirical measure, and the function $f$ is convex with a special structure. 
We refer to \citet{Zhao-Guan} for a similar setting and another class of objective functions $f$.
\citet{Bla-Mur} and \citet{Gao-Kle} obtain the result on a general Polish space and for lower semicontinuous cost functions $c$ and upper semicontinuous integrable functions $f$.
The proof given in the present paper is essentially more direct.
It relies on a version of Choquet's capacitability theorem and applies for measurable functions.

The application of this duality goes well beyond dimension reduction for robust optimized certainty equivalents.
It also enables us to derive a finite dimensional representation of the robust shortfall risk measure of \citet{foellmer02}.
Interestingly, a \emph{``martingale version''} of \eqref{eq:rep.ball_intro} provides a finite dimensional analytical formula for robust European option pricing, see Proposition \ref{prop:pricing.dual}.
As an example, if we consider the price $\mathrm{CALL}(k)$ of a call option on a stock with time one price $S$ and time zero price $s$ and risk neutral distribution $\mu_0$, for the robust price of a linearly penalized call, we obtain
\begin{equation*}
    \begin{split}
        \mathcal{CALL}(k) := & \sup_{\left\{\mu \in \mathcal{M}_1(\mathbb{R}^d)\colon   \int_{\mathbb{R}^d} S \,d\mu= s\right\}} \left( \int (x-k)^+\,d\mu-d_c(\mu_0,\mu)\right)\\
        =&\inf_{\alpha\in\mathbb{R}}\left( \mathrm{CALL}\left( k-\frac{2\alpha +1}{2}\right)+\frac{\alpha^2}{2}  \right)
    \end{split}
\end{equation*}
when $d_c$ is the second order Wasserstein distance.
In other terms, robust pricing of European options with respect to penalization of a Wasserstein distance simply boils down to optimizing the payoff -- here the strike -- of the option price with respect to the baseline risk neutral distribution.

Coming back to the optimized certainty equivalent, we further investigate alternative representations of the robust OCE when the nonlinear expectation ${\cal R}$ is the worst case expectation over an arbitrary set ${\cal D}\subseteq {\cal M}_1(\mathbb{R})$. 
In particular, we derive the representation of the robust $AV@R$ as a tail risk measure.
This representation, first proved by \citet{Roc-Ury} in the non-robust case, has important applications in optimization problems and does not carry over to the robust case unless stronger structural and topological assumptions are put on the set ${\cal D}$, see Proposition \ref{prop:avar-int}.
Under such assumptions, we derive a Kusuoka-type representation (see \citet{Jou-Sch-Tou06} and \citet{Jou-Sch-Tou06}) for ``robustifications'' of law-invariant risk measures, see Corollary \ref{cor:kusuoka}.
Finally, when defined on random variables, convex dual representations of OCE are particularly relevant, see for instance \citet{kupper01} for applications to optimization problems and \citet{OCE-PDE} for applications to dynamic representations.
In the present article, we derive a dual representation of the robust OCE on the set of bounded measurable random variables, without any topological assumption on the sample space.

Financial modeling under model ambiguity, known as robust finance, is currently the subject of intensive research.
Earliest papers dealing with robust risk measures with the ambiguity set given by a Wasserstein ball include \citet{Pflug-Pich-Woz} and \citet{Wozabal14} in the context of portfolio optimization problems.
Among other more recent papers, we refer to \cite{Hob98,robhedging,bei-hl-pen,DS_transport} for superhedging problems and 
\cite{Bartl_Utilitymax,carassus2016robust,Mat-Pos-Zhou15, Nutz_Robutility16, Neu-Sik} for robust utility maximization.
Distributionally robust problems are also studied in statistics, economics and operations research, 
see for instance \cite{huber,Han-Sar01}.

The paper is organized as follows: The next section summarizes our main findings. 
Namely, a duality result for the nonlinear operator ${\cal R}$ for general penalty functions and measurable cost functions.
This result is then used to provide finite dimensional representations of the OCE, shortfall risk measure and price of European options in the presence of model ambiguity.
Furthermore, we give conditions under which the robust $AV@R$ is a tail-risk measure, and a convex dual representation when the OCE is defined on random variables.
We also give several edifying examples.
In the final section, we provide detailed proofs.

\section{Main results for uncertainty given by Wasserstein distances} 
\label{sec:main.results}    
We start this section with briefly defining our notation, a short review on the distances emerging from optimal transport, and the duality theorem hinted in \eqref{eq:rep.ball_intro}.
Throughout, let $d \in \mathbb{N}$ and $X\subseteq \mathbb{R}^d$ be a closed set.
We denote by $\mathcal{M}_1(X)$ the set of all probabilities on the Borel $\sigma$-field of $X$.
For a measurable function $f\colon X\to(-\infty,\infty]$ bounded from below and $\mu\in\mathcal{M}_1(X)$, we denote by $\int f\,d\mu=\int_X f(x)\,\mu(dx)\in(-\infty,\infty]$ the integral of $f$ against $\mu$.
Now fix some measurable function $c\colon X\times X \to[0,\infty]$ and define the cost of transportation by
\begin{equation*}
    d_c(\mu,\nu):=\inf\left\{ \int\nolimits_{X\times X} c(x,y)\,\pi(dx,dy) : 
        \begin{array}{l}
            \pi\in\mathcal{M}_1(X\times X) \text{ such that}\\
            \pi(\cdot\times X)=\mu \text{ and } \pi(X \times\cdot)=\nu
        \end{array}
\right\}
\end{equation*}
for $\mu,\nu\in\mathcal{M}_1(X)$.
If $c$ is assumed to be lower semicontinuous, this problem has a dual formulation
\begin{align}
    \label{eq:dual.for.transport}
    d_c(\mu,\nu)=
    \sup\left\{ \int\nolimits_X f\,d\mu +\int\nolimits_X g\,d\nu :
        \begin{array}{l}
            f,g\colon X\to \mathbb{R}\text{ bounded continuous such that}\\
            f(x)+g(y)\leq c(x,y)\text{ for all }x,y\in X
    \end{array} \right\},
\end{align}
see for instance \cite[Theorem 5.9]{villani2008optimal} for a proof.
\begin{remark}
    For the specific choice $c(x,y)=|y-x|^p$ with $p\geq 1$, the function $d_c^{1/p}$ defines the celebrated Wasserstein distance of order $p$. In case of compact $X$, this is in fact a metric on $\mathcal{M}_1(X)$ compatible with the weak topology, that is the coarsest topology making all mappings $\mu\mapsto \int f\,d\mu$ continuous for continuous and bounded $f$.
    However, for general $X\subset\mathbb{R}^d$, one has $d_c(\mu_n,\mu)\to 0$ if and only if $\mu_n$ converges weakly to $\mu$ and $\int_X |x|^p\,\mu_n(dx)\to \int_X |x|^p\,\mu(dx)$, see \cite[Theorem 6.8]{villani2008optimal}.
\end{remark}
For a function $f\colon X\to(-\infty,\infty]$ and $\lambda\geq 0$ it's $\lambda c$-transform is defined and denoted by
\begin{equation*}
   f^{\lambda c}(x):=\sup\left\{ f(y)-\lambda c(x,y) : y\in X \text{ such that } f(y)<\infty\right\}. 
\end{equation*}
\begin{remark}
    If $f$ and $c$ are continuous, then $f^{\lambda c}$ is lower semicontinuous.
    In general, if $f$ is only assumed to be measurable, then $f^{\lambda c}$ is not necessarily.
    However, if $f$ is measurable, it follows for instance from \cite[Section 7]{bertsekas1978stochastic} that $f^{\lambda c}$ is universally measurable and in particular $\mu$-measurable for every $\mu\in\mathcal{M}_1(X)$.
    This implies that the integral $\int f^{\lambda c}\,d\mu$ is well defined whenever $f$ is measurable and bounded from below.
    Note also that $f^{\lambda c}$ is a well-known modification of the classical Fenchel-Legendre transform studied in the context of optimal transport under the name ``$c$-transform'', see for instance \cite[Section 5]{villani2008optimal}.
\end{remark}

From now on, we call
\begin{itemize}
    \item a \emph{cost function} any lower semicontinuous function $c\colon X\times X \to [0,\infty]$ with $\inf_{y\in X}c(x,y)=0$ for all $x\in X$, and for every $r\geq 0$ there is $k\geq 0$ such that $c(x,y)\geq r$ if $|x-y|\geq k$;
    \item a \emph{penalization function} any convex, increasing lower semicontinuous function $\varphi\colon [0,\infty]\to[0,\infty]$ with $\varphi(0)=0$ and neither $\varphi$ nor $\varphi^\ast$ being constant $0$.\footnote{$\varphi^\ast$ denotes the convex conjugate of $\varphi$, that is, $\varphi^\ast(y)=\sup_x(xy-\varphi(x))$.}
\end{itemize}
\begin{remark}
    For the typical case when the cost function depends only on the difference -- with slight abuse of notations $c(y-x)=c(x,y)$ -- the assumptions of a cost function are satisfied whenever $\inf_{x\in X} c(x)=0$ and $\liminf_{|x|\to\infty} c(x)=\infty$.
    In particular, $c(x,y):=|y-x|^p$ with $p>0$ corresponding to the Wasserstein distances are all cost functions.
\end{remark}
\begin{theorem}
    \label{thm:integral.neigborhood}
    Given a closed set $X\subset \mathbb{R}^d$, a cost function $c$ and a penalization function $\varphi$, then it holds that
    \begin{equation} 
        \label{eq:rep.ball}
        \sup_{\mu \in\mathcal{M}_1(X)}  \left( \int\nolimits_X f\,d\mu - \varphi\left(d_c\left(\mu_0,\mu\right)\right)\right)=
        \inf_{\lambda\geq 0} \left( \int\nolimits_X f^{\lambda c}\,d\mu_0 +\varphi^\ast(\lambda) \right)
    \end{equation}
    for every measurable function $f\colon X\to(-\infty,\infty]$ bounded from below.
    Moreover, the infimum over $\lambda$ is attained.
\end{theorem}
\begin{example}
    Typical examples of penalization function $\varphi$ we have in mind:
    \begin{itemize}[fullwidth]
        \item In case of $\varphi=\infty 1_{(\delta,\infty]}$ for some fixed $\delta>0$, one computes $\varphi^\ast(\lambda)=\delta\lambda$ so that
            \begin{equation}
                \label{eq:integral.formula.neighborhood}
                \sup_{\left\{\mu \in\mathcal{M}_1(X) \colon d_c(\mu_0,\mu))\leq\delta\right\}}  \int f\,d\mu =\inf_{\lambda\geq 0} \left( \int f^{\lambda c}\,d\mu_0 + \delta\lambda \right).
            \end{equation}
        \item
            For $\varphi(x)=x$ one gets $\varphi^\ast=\infty 1_{(1,\infty)}$ and as $f^{\lambda c}\leq f^{c}$ for all $\lambda\leq 1$, the formula simplifies to
            \begin{equation*}
                \sup_{\mu \in\mathcal{M}_1(X)}  \left( \int f\,d\mu - d_c(\mu_0,\mu))\right)=\int f^{c}\,d\mu_0.  
            \end{equation*}
        \item
            Other examples might include $\varphi(x)=x^p/p$ for some $p>1$ for which $\varphi^\ast(\lambda)=\lambda^q/q$ where $1/p+1/q=1$, or 
            $\varphi(x)=\exp(x)-1$ for which $\varphi^\ast(\lambda)=\lambda\log\lambda-\lambda+1$.
    \end{itemize}
\end{example}
\begin{remark}
    Note that the formula \eqref{eq:integral.formula.neighborhood} has recently been proven by several authors, under different assumptions.
    In \cite{Esfahani} and \cite{Zhao-Guan}, the authors focus on the case $c(x,y)= |x-y|$ and $\mu_0$ an empirical measure.
    The closest set of assumptions to ours is in \cite{Bla-Mur}.
    Therein, the authors work on a general Polish space $X$, assume $c$ to be lower semicontinuous and real-valued, and prove duality for ($\mu_0$-integrable) upper semicontinuous functions $f$. See also \cite{Gao-Kle}.
    Further, note that the techniques differ among all proofs. 
    Our proof builds on convex and Choquet's regularity result for functional defined on measurable functions and is significantly shorter.
\end{remark}

\subsection{Optimized certainty equivalents and Shortfall risk measures}

Throughout this section let $X=\mathbb{R}$ and $\mu_0\in\mathcal{M}_1(\mathbb{R})$ be some fixed baseline distribution. 
For any ``loss'' function $l\colon\mathbb{R}\to(-\infty,\infty]$ measurable and bounded from below, recall that the optimized certainty equivalent and the shortfall risk measure with respect to $\mu_0$ are defined by
\begin{equation*}
   \mathrm{OCE}(l) :=\inf_{m\in\mathbb{R}}  \left( \int l(\cdot-m))\,d\mu_0 +m  \right) \quad\text{and}\quad \mathrm{ES}(l) :=\inf\left\{ m\in\mathbb{R} : \int l(\cdot-m) \,d\mu_0\leq 0 \right\}. 
\end{equation*}
For the remainder we fix a cost function $c$, which for notational simplicity depends only on the difference, and a penalization function $\varphi$.
Given $f\colon\mathbb{R}\to(\infty,\infty]$ measurable and bounded from below, define
\begin{equation*}
    \mathcal{R}(f):= \sup_{\mu\in\mathcal{M}_1(\mathbb{R})}\left( \int f\,d\mu- \varphi\left( d_c(\mu_0,\mu)) \right)\right). 
\end{equation*}
and the robust optimized certainty equivalent / shortfall risk measure
\begin{equation*}
    \mathcal{OCE}(l) :=\inf_{m\in\mathbb{R}}  \left( \mathcal{R}(l(\cdot-m)) +m  \right) \quad\text{and}\quad \mathcal{ES}(l) :=\inf\big\{ m\in\mathbb{R} : \mathcal{R}(l(\cdot-m))\leq 0 \big\}
\end{equation*}
for $l\colon\mathbb{R}\to(-\infty,\infty]$ measurable and bounded from below.

The following is the first main result of this section which states that the infinite dimensional problem of computing the quantities $\mathcal{OCE}(l)$ and $\mathcal{ES}(l)$ is in fact a finite dimensional problem with different loss function $l$.

\begin{theorem}
    \label{thm:oce.and.es.robust.compute}
    Given a loss function $l\colon\mathbb{R}\to\mathbb{R}$ measurable and bounded from below, 
    it holds that
    \begin{equation*}
        \mathcal{OCE}(l)=\inf_{\lambda\geq0} \left( \mathrm{OCE}\left(l^{\lambda c}\right) + \varphi^\ast(\lambda)\right).
    \end{equation*}
    If further $\inf_x l(x)<0$, then it holds that
    \begin{equation*}
       \mathcal{ES}(l) =\inf_{\lambda\geq0} \mathrm{ES}\left(l^{\lambda c}+\varphi^\ast(\lambda)\right). 
    \end{equation*}
\end{theorem}

\begin{remark}
    \label{rem:special.case.varphi.oce}
    In the special case of $\varphi(x)=x$, computations simplify and
    \begin{equation*}
        \mathcal{OCE}(l) =\mathrm{OCE}\left(l^{c}\right) \quad\text{and}\quad \mathcal{ES}(l) =\mathrm{ES}\left(l^{c}\right).    
    \end{equation*}
    Similar, if $\varphi(x)=\infty1_{(\delta,\infty]}$ for some $\delta>0$, one gets 
    \begin{equation*}
        \mathcal{OCE}(l)=\inf_{\lambda\geq0} ( \mathrm{OCE}(l^{\lambda c}) + \lambda \delta)  \quad \text{and}\quad \mathcal{ES}(l) =\inf_{\lambda\geq0} \mathrm{ES}(l^{\lambda c}+\lambda\delta).
    \end{equation*}
\end{remark}
\begin{remark}
    \label{rem:c.bounded.or.different.metrik}
    The main reason why the Wasserstein distance is a popular choice of distance to model ambiguity is that the empirical measure converges with respect to this distance, and it is not too strong as opposed to, for instance, the total variation distance, see for instance \cite{dereich2013constructive}.
    For example, the distance between the true measure and the empirical measures converges in expectation, with non-asymptotic rates, roughly in the order $n^{-1/2}$ if enough moments exist.
    This justifies the use of optimal-transport distances to model the ambiguity set in finance or any other field where the true distribution is approximated by the empirical measure built on available data.
    The convergence implies that the approximation becomes more accurate as the data sample size increases.
    Moreover, concentration inequalities suggest how to choose $\delta$ in terms of $N$.
    More precisely, by \cite{fournier2015rate}, if $c(x,y)=|x-y|^p$ and the true distribution $\mu$ is assumed to satisfy $\int \exp( |x|^{2p})\,\mu(dx)<\infty$, then it holds that
    \begin{equation*}
    P(d_c(\mu,\mu_N)\ge \delta)\leq C\exp(-cN\delta^2)\quad\text{for all }\delta\in(0,\infty)
    \end{equation*}
     where $\mu_N:=\frac{1}{N}\sum_{n=1}^N\delta_{x_n}$ is the empirical measure on $N$ 
     i.i.d.~observations $x_i,\dots,x_N$ and $C,c>0$ are constants
    (the strong exponential integrability assumption can be weakened to existence of moments,
    but the formula gets uglier).   Also refer to \cite{Bol-Gui-Vil} for similar bounds.

    There are many other distances on the space of probability measures one could take into account when defining $\mathcal{R}$, and consequently $\mathcal{OCE}$ or $\mathcal{ES}$ (see for instance \cite{gibbs2002choosing} for a survey).
    We already noted before that the Wasserstein distance is stronger than weak convergence, which actually turns out to be necessary in the present setting:
    If one replaces $d_c$ by a distance compatible with weak convergence or if $\liminf_{x\to\infty} c(x)<\infty$ and $\mu_0 = \delta_0$, then always $\mathcal{OCE}(l)=\mathcal{ES}(l)=\infty$ as soon as $l$ is a convex and not constant loss function.
    A proof of these facts is given in Section \ref{sec:proofs}.
\end{remark}

While Theorem \ref{thm:oce.and.es.robust.compute} reduces the infinite dimensional problem to a finite dimensional one, regardless of the computation of the classical OCE, the challenge of computing $l^{\lambda c}$, usually for several different $\lambda$'s, remains.
However, for many relevant examples closed form formulas exist.

\begin{example}[Average value-at-risk]
    \label{ex:avar.formula}
    Let $l(x)=x^+/\alpha$ for some $\alpha\in(0,1)$ so that $\mathcal{OCE}(l)$ becomes the robust average value-at-risk $\mathcal{AV@R}_\alpha$ at level $\alpha$.
    Table \ref{t:AVAR} summarizes the relation between the non-robust average value-at-risk $\mathrm{AV@R}$ and its robust counterpart $\mathcal{AV@R}$ for different choices of $c$ and $\varphi$.
    \begin{table}[H]
        \begin{center}
            \begin{tabular}{@{}lclcl@{}}
                \toprule
                $\mathcal{AV@R}_\alpha$ & & $c(x,y)=|x-y|$ & & $c(x,y)=|x-y|^2$\\
                \midrule
                $\varphi=\infty1_{(\delta,\infty]}$ & & $\mathrm{AV@R}_\alpha + \delta/\alpha$ & &  $\mathrm{AV@R}_\alpha + (\delta/\alpha)^{1/2}$ \\
                $\varphi(x)=x$ & & $\infty$ & & $\mathrm{AV@R}_\alpha+1/(4\alpha)$\\
                $\varphi(x)=x^p/p$ & & $\mathrm{AV@R}_\alpha + (1/\alpha)^{q}/q$ & & $\mathrm{AV@R}_\alpha+ (p+1)/(p (4\alpha)^{-p/(p+1)})$ \\
                \bottomrule 
            \end{tabular}
            \caption{Robust average value-at-risk for selected penalization functions $\varphi$ and Wasserstein distances.
            Here $q>1$ is such that $1/p+1/q=1$.}
            \label{t:AVAR}
        \end{center}
    \end{table}

\end{example}
This example gives a mathematical justification to an intuitively natural fact known as post-valuation adjustment.
When computing the risk of a loss $\mu_0$, it is advisable to add a margin to hedge a possible model misspecification or a computational error, see for instance \cite[Chapter 5]{Damodaran}.

\begin{example}[Monotone mean-variance]
    \label{ex:mean-var}
    Let $l(x) = (((1 + x)^+)^2 -1)/2 $ so that $\mathcal{OCE}(l)$ becomes the robust monotone mean-variance risk measure, see \cite{MMR2006}.
    For the cost function $c(x,y) = (x-y)^2$, one has
    \begin{align*}
        \mathcal{OCE}(l)
        =\inf_{\lambda>1/2}\left( \mathrm{OCE}\left(\frac{2\lambda}{2\lambda-1} l \right) 
        +\frac{1}{4\lambda-2} - \varphi^\ast(\lambda) \right).
    \end{align*}
    For example if $\varphi(x)=x$, this formula simplifies to 
    $\mathcal{OCE}(l) = \mathrm{OCE}(2l) + 1/2$.
\end{example}

\begin{example}[Value-at-risk]
    \label{ex:var}
    Let $c(x,y)=|x-y|^p$ for some $p>0$.
    Then, for the robust value-at-risk at level $\alpha\in(0,1)$, one has
    \begin{align*}
        \mathcal{V@R}_\alpha
        &:=\inf\big\{m\in\mathbb{R} : \mathcal{R}(1_{(m,\infty)})\leq\alpha\big\}\\
        &=\inf_{\lambda\geq0}\inf\big\{ m\in\mathbb{R} : \mu_0((m,\infty))+ e(m,\lambda)+\varphi^\ast(\lambda) \leq\alpha\big\},
    \end{align*}
    where $e(m,\lambda):=\int_{(m-\lambda^{-1/p},m]}1-\lambda|x-m|^p\,\mu_0(dx)$.
    For example if $\varphi(x)=x$ and $c(x,y)=|x-y|$ this formula simplifies and $\mathcal{V@R}_\alpha=\inf\{ m\in\mathbb{R} : \mu_0((m,\infty))+\int_{(m-1,m]}(m-x)\,\mu_0(dx)\leq\alpha\}$.
\end{example}

\subsection{Robust pricing of European options}

The principle behind Theorem \ref{thm:oce.and.es.robust.compute} is not limited to the OCE or ES, or other risk measures of similar form.
It can be applied for example to option pricing, see also the recent paper \cite{Bla-Chen-Zho} for applications to mean variance hedging.

Throughout let $X=\mathbb{R}^d$ be the canonical space of a $d$-dimensional finance asset $S=(S^1,\ldots,S^d)$, that is, $S\colon\mathbb{R}^d\to\mathbb{R}^d$ is the identity.
Here again, we fix a cost function $c$, which for notational simplicity depends only on the difference, and a penalization function $\varphi$.
Further we assume that for some $\varepsilon>0$, it holds $\liminf_{|x|\to\infty} c(x) / |x|^{1+\varepsilon} = \infty$.
Let $\mu_0\in\mathcal{M}_1(\mathbb{R}^d)$ be an integrable risk neutral pricing measure for $S$, that is $s=\int S\,d\mu_0\in\mathbb{R}^d$ is the price of these assets at time $0$
and assume that $\int c(x-y)\,\mu_0(dx)<\infty$ for every $y\in\mathbb{R}^d$.
We further assume that the interest rate is $0$.

Given a European type of option payoff $H:= h(S)$ where $h:\mathbb{R}^d\to \mathbb{R}$ is measurable and bounded from below, we denote by
\begin{equation*}
    \mathrm{PRICE}(H):=\int\nolimits_{\mathbb{R}^d} h(x)\,\mu_0(dx)= \int h(S)\,d\mu_0.
\end{equation*}
its risk neutral price.
Taking uncertainty in the pricing measure into account, an analogue of $\mathcal{R}$ in the previous section, consists of considering all probabilities consistent with the asset prices, that is,
\begin{equation*}
    \mathcal{PRICE}(H):=\sup_{\left\{\mu \in \mathcal{M}_1(\mathbb{R}^d)\colon \int_{\mathbb{R}^d} S \,d\mu= s\right\}}\left( \int h\,d\mu-\varphi\left(d_c(\mu_0,\mu))\right)\right).
\end{equation*}
As the additional constraint $\int S \,d\mu= s$ is satisfied if and only if $\inf_{\alpha\in\mathbb{R}^d}\int \alpha\cdot (S-s) \,d\mu>-\infty$, formally applying a minimax theorem one obtains

\begin{proposition}
    \label{prop:pricing.dual}
    For every measurable payoff $h\colon\mathbb{R}^d\to\mathbb{R}$ such that $\sup_{x\in\mathbb{R}^d} |h(x)|/(1+|x|)<\infty$ one has
    \begin{equation*}
       \mathcal{PRICE}(H)=\inf_{\alpha \in \mathbb{R}^d}\inf_{\lambda \geq 0}\left( \int h^{\lambda c, \alpha}\,d\mu_0+\varphi^\ast(\lambda)\right) 
    \end{equation*}
    where
    \begin{equation*}
        h^{\lambda c, \alpha}(x):=\sup_{y \in \mathbb{R}^d} \left( h(y)+\alpha\cdot(y-s)-\lambda c(y-x) \right) \quad \text{for } x \in \mathbb{R}^d, \alpha \in \mathbb{R}^d, \lambda \geq 0.
    \end{equation*}
\end{proposition}

The latter function $h^{\lambda c, \alpha}$ can be interpreted as a modified payoff priced against the original risk neutral pricing distribution $\mu_0$.
In particular if $\varphi(x)=x$ the formula again simplifies to
\begin{equation*}
    \mathcal{PRICE}(H)=\inf_{\alpha \in \mathbb{R}^d}\int h^{c, \alpha}\,d\mu_0.
\end{equation*}

\begin{example}[Robust Call]
    \label{ex:call.price}
    Let us consider the case of a call option with maturity $1$ and strike $k$ on a single asset, that is, $h(x)=(x-k)^+$ and $d=1$.
    For ease of notations, we denote by $\mathrm{CALL}(k)$ the corresponding price.
    For the cost $c(x,y)=(x-y)^2/2$, the robust call at strike $k$ satisfies
    \begin{equation*}
       \mathcal{CALL}(k)=\inf_{\alpha\in\mathbb{R}}\inf_{\lambda > 0}\left( \mathrm{CALL}\left( k-\frac{2\alpha +1}{2\lambda}\right)+\frac{\alpha^2}{2\lambda}+\varphi^\ast(\lambda)\right). 
    \end{equation*}
    Again, if $\varphi(x)=x$, the robust price of a call simplifies to
    \begin{equation*}
        \mathcal{CALL}(k)=\inf_{\alpha\in\mathbb{R}}\left( \mathrm{CALL}\left( k-\frac{2\alpha +1}{2}\right)+\frac{\alpha^2}{2} \right).
    \end{equation*}
    Figure \ref{f:callprice} represents the standard Black and Scholes price versus robust price as a function of the strike.
    As expected, the largest spread between the Black and Scholes price and the robust one is at the money, while the cost of robustness vanishes for in or out of the money options.
    \begin{figure}[H]
        \centering
        \includegraphics[width=0.9\textwidth]{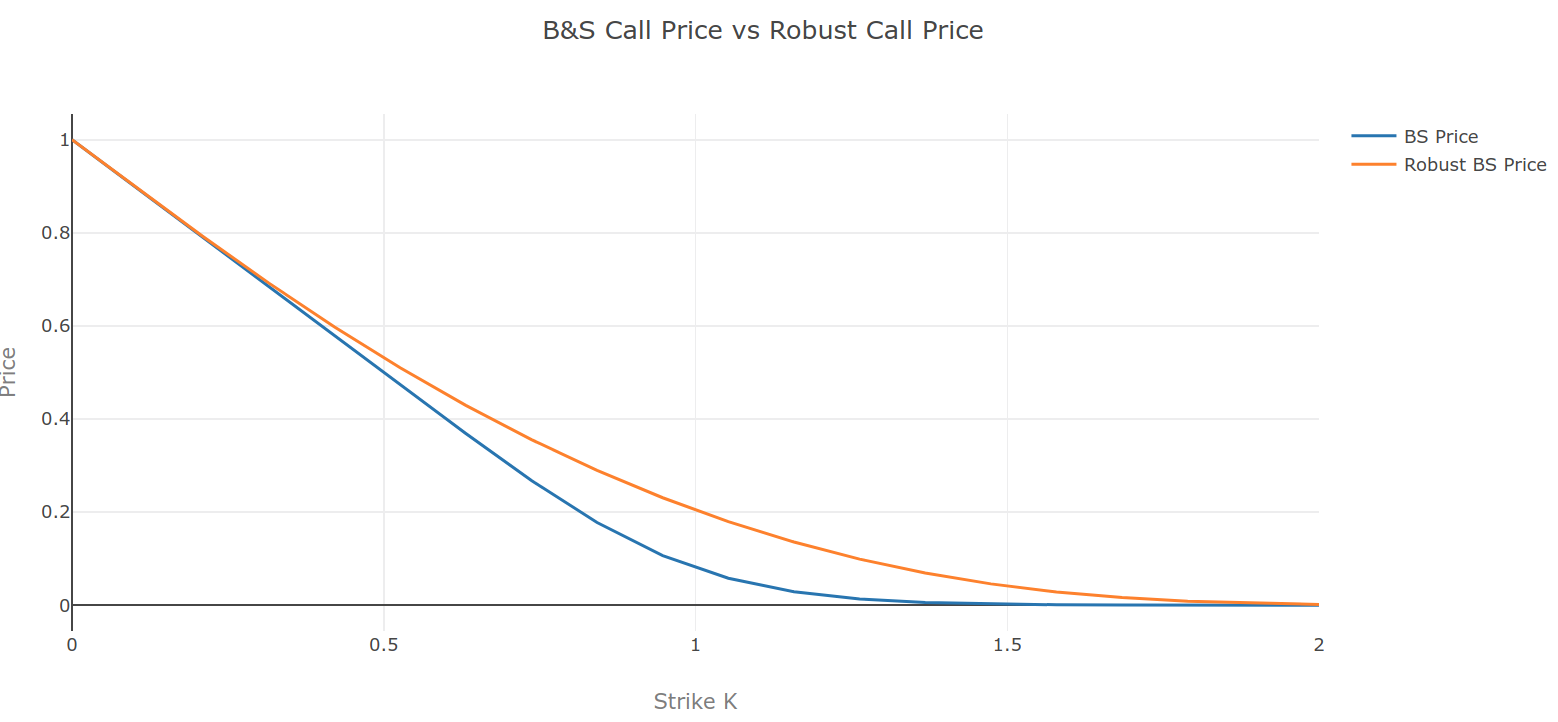} 
        \caption{Call Option price, Black and Scholes vs.~Robust. $\varphi(x)=x$, spot price is one, zero interest rate, the distribution is log-normal with $20\%$ volatility and one year maturity.}
        \label{f:callprice}
    \end{figure}
\end{example}
\begin{remark}
    The robust price $\mathcal{PRICE}(H)$ can also be interpreted as the minimal superhedging price of $H$ when the shortfall risk is controlled by the nonlinear expectation $\mathcal{R}$.
    More precisely, \cite{robhedging} give conditions under which the robust price can be represented as
    \begin{equation*}
        \inf\{m\in \mathbb{R}:\mathcal{R}(-m- \alpha\cdot(S-s)+H)\le 0 \text{ for some } \alpha \in \mathbb{R}^d\}.
    \end{equation*}
\end{remark}

\begin{remark}[Robust Utility maximization]
    Another problem in robust mathematical finance where Theorem \ref{thm:integral.neigborhood}
    directly applies is utility maximization.
    Indeed, given a utility function $U\colon\mathbb{R}\to\mathbb{R}$ bounded from above
    and a measurable claim $f\colon\mathbb{R}^d\to\mathbb{R}$,
    Theorem \ref{thm:integral.neigborhood} implies that
    \[ \sup_{\alpha\in\mathbb{R}^d} \inf_{\mu\in\mathcal{M}_1(\mathbb{R}^d)}
    \left(\int U(f(S) + \alpha\cdot (S-s))\,d\mu +\varphi(d_c(\mu_0,\mu)) \right)
    =\sup_{\alpha\in\mathbb{R}^d}\sup_{\lambda\geq 0}  
    \left( \int U^{\lambda c, \alpha}\,d\mu_0  - \varphi^\ast(\lambda)\right) ,\]
    where $U^{\lambda c,\alpha}(x):=\inf_y (U(f(y) + \alpha\cdot (y-s)) +\lambda c(x,y))$.
    Note that one can also treat the multi-period case with dynamic programming,
    see \cite[Section 2.3]{Bartl_Utilitymax} for a discussion of the Wasserstein distance
    in this framework.
\end{remark}

\section{Results for general sets}
\label{seq:kusuoka.and.duality}

\subsection{Tail risk measures and Kusuoka type representation}
\label{sec:directed}

For this section, let  
\begin{equation*}
   \mathcal{D}\subseteq\mathcal{M}_1(\mathbb{R}) \quad \text{and} \quad
   \mathcal{R}(f):=\sup_{\mu\in\mathcal{D}} \int f\,d\mu
\end{equation*}
 for $f\colon\mathbb{R}\to(-\infty,\infty]$ measurable and bounded from below.
In the non-robust setting -- that is, $\mathcal{D}$ being a singleton --  it is well known that the average value-at-risk, see Example \ref{ex:avar.formula}, is a risk measure capturing the ``tail risk'' by satisfying the representation
\begin{align}
    \label{eq:avar.is.integral.var}
    \mathrm{AV@R}_\alpha=\frac{1}{\alpha}\int\nolimits_{0}^\alpha \mathrm{V@R}_u\,du,
\end{align}
where $\mathrm{V@R}$ is the value-at-risk, see Example \ref{ex:var}.
That is, $\mathrm{AV@R}$ is roughly speaking the average over the $\mathrm{V@R}$ below the $\alpha$-quantile;
an important property for instance in optimal portfolio problems, see \citet{Roc-Ury}.
However, we will see in Section \ref{sec:dual} that $\mathcal{AV@R}_\alpha$ equals the supremum over $\mu\in\mathcal{D}$ of the average value at risk with respect to $\mu$, from which it easily follows that \eqref{eq:avar.is.integral.var} in general no longer holds true when $\mathcal{D}$ consists of more than one element. 
In a similar manner, one cannot expect any form of Kusuoka type representation -- abstract version of \eqref{eq:avar.is.integral.var} where any law invariant risk measure can be represented over the average value at risk -- in general.

In order to prove a robust version of formula \eqref{eq:avar.is.integral.var}, stronger assumptions on the set ${\cal D}$ are needed.
If $\mathcal{D}\neq\emptyset$ satisfies
\begin{align}\label{eq:ass.dir}\tag{DIR}
    \begin{array}{l}
        \text{$\mathcal{D}$ is tight and for all $\mu,\tilde{\mu}\in\mathcal{D}$ there is $\nu\in\mathcal{D}$}\\
        \text{such that $\mu(-\infty,t],\tilde{\mu}(-\infty,t]\ge \nu(-\infty,t]$ for all $t$}
\end{array}\bigg\},
\end{align}
it also follows that the robust $\mathrm{OCE}$ has the same properties as the non-robust one, see Corollary \ref{cor:when_D_directed}.
Here tight means that for every $\varepsilon>0$ there is some compact set $K\subset\mathbb{R}$ satisfying $\sup_{\mu\in\mathcal{D}}\mu(K^c)\leq\varepsilon$.
\begin{proposition}
    \label{prop:avar-int}
    Assume that \eqref{eq:ass.dir} holds true. Then
    \begin{equation*}
        \mathcal{AV@R}_\alpha = \frac{1}{\alpha}\int\nolimits_{0}^\alpha \mathcal{V@R}_u\,du
    \end{equation*}
    for every $\alpha\in(0,1)$.
\end{proposition}

\begin{example}
    \label{exa:dir}
    Let $(\Omega, {\cal F}, P)$ be a probability space carrying a Brownian motion $(W_t)_{t\in [0,T]}$,
    where $T\in (0,\infty)$,
    equipped with the completion of the natural filtration of $W$.
    Let $\sigma$, $\ubar{b}$, and $\bar{b}$ be three real numbers such that $\sigma>0$.
    Then, for every strictly increasing function $f\colon\mathbb{R}\to\mathbb{R}$ and $S_0>0$, the set
    \begin{equation*}
        \mathcal{D}:=\big\{P\circ f(S_T)^{-1} \text{ with } dS_t = S_t(b\,dt + \sigma\,dW_t)\text{ and } 
        b\in[\ubar{b},\bar{b}] \big\}
    \end{equation*}
    satisfies \eqref{eq:ass.dir}.
\end{example}

Proposition \ref{prop:avar-int} suggests a Kusuoka type representation for robustifications of law-invariant risk measures.
In fact, let $\rho$ be the risk measure defined as
\begin{equation*}
    \rho(\mu):= \sup_{\nu \in {\cal M}_1((0,1])}
    \left( \int\nolimits_{(0,1]}AV@R_u(\mu)\,\nu(du) - \beta(\nu) \right)
\end{equation*}
for some penalty function $\beta\colon \mathcal{M}_1((0,1])\to(-\infty,\infty]$,
and consider its robust counterpart given by
$ \rho({\cal D}) := \sup_{\mu \in {\cal D}}\rho(\mu)$.
The following corollary gives a representation of $\rho({\cal D})$ in terms of ${\cal AV@R}$.
In particular, it shows that when ${\cal D}$ satisfies \eqref{eq:ass.dir}, ${\cal AV@R}$ constitute basic building blocks of robustifications of law-invariant risk measures.
\begin{corollary}
\label{cor:kusuoka}
    Assume that ${\cal D}$ satisfies \eqref{eq:ass.dir} and is closed in the weak topology.
    Then, it holds that
    \begin{equation*}
    \rho({\cal D}) := \sup_{\mu \in {\cal D}}\rho(\mu)
    = \sup_{\nu \in {\cal M}_1((0,1])}
    \left(\int\nolimits _{(0,1]}{\cal AV@R}_u({\cal D})\,\nu(du) - \beta(\nu) \right).
    \end{equation*}
\end{corollary}
\subsection{Duality}
\label{sec:dual}

Let $(\Omega,\mathcal{F})$ be a given measurable space endowed with a non-linear expectation 
\begin{equation*}
    \mathcal{E}(\cdot):=\sup_{P\in\mathcal{M}_1(\Omega)} \left(E_P[\cdot]-\beta(P)\right),
\end{equation*}
for some function $\beta \colon {\cal M}_1(\Omega)\to [0,\infty]$, where ${\cal M}_1(\Omega)$ is the set of probability measures on $\mathcal{F}$.
In analogy to the first part of the paper, we define the robust OCE as
\begin{equation*}
    \mathcal{OCE}(X)=\inf_{m\in\mathbb{R}}\left(\mathcal{E}(l(X-m))+m\right)   
\end{equation*}
for every measurable function $X\colon\Omega\to\mathbb{R}$.
Here $l\colon\mathbb{R}\to\mathbb{R}$ is assumed to satisfy the usual assumptions
\begin{align}
    \label{eq:ass.l}\tag{CIB}
    \begin{array}{l}
        l \text{ is convex, increasing, bounded from below, and } \\
        l(0)=0, \,l^\ast(1)=0,\text{ and } l(x)>x \text{ for $|x|$ large enough}
\end{array}\bigg\}
\end{align}
and $l^*(y)=\sup_{x \in \mathbb{R}}(xy-l(x))$ denotes the convex conjugate of $l$ for $y\in \mathbb{R}$ and $l^*(\infty):=\infty$.
Note that $l^\ast(y)\geq 0$ and that if $l$ is continuously differentiable, 
then $l^\ast(1)=0$ just says that $l'(0)=1$.
For the remainder of this section $dQ/dP$ denotes the Radon-Nikodym derivative if
$Q$ is absolutely continuous with respect to $P$ and $dQ/dP\equiv\infty$ otherwise.

\begin{theorem}
\label{thm:dual}
    Assume that $l$ satisfies \eqref{eq:ass.l} and that $\beta$ is convex with $\inf_{P \in \mathcal{M}_1(\Omega)}\beta(P)=0$.
    Then one has 
    \begin{equation*}
        \mathcal{OCE}(X)=\sup_{Q\in\mathcal{M}_1(\Omega)} \left( E_Q[X] - \inf_{P\in\mathcal{M}_1(\Omega)} \left(E_P\left[l^\ast\left(\frac{dQ}{dP}\right) \right]+\beta(P)\right)\right)
    \end{equation*}
    for every bounded measurable function $X\colon\Omega\to \mathbb{R}$.
\end{theorem}

When $\beta$ is the convex indicator of a non-empty convex subset ${\cal P}$ of ${\cal M}_1(\Omega)$, that is $\beta=\infty 1_{\mathcal{P}^c}$, then $\mathcal{E}(\cdot)=\sup_{P\in\mathcal{P}} E_P[\cdot]$ and 
\begin{equation*}
    \mathcal{OCE}(X)=\sup_{Q\in\mathcal{M}_1(\Omega)} \left( E_Q[X] -\inf_{P\in\mathcal{P}} E_P\Big[l^\ast\left(\frac{dQ}{dP}\right) \Big]\right).
\end{equation*}
In particular, this implies that the robust optimized certainty equivalent is equal to the supremum over $P\in\mathcal{P}$ of the optimized certainty equivalent with respect to $P$.

\begin{example}
    Troughout this example let $\mathcal{E}(\cdot)=\sup_{P\in\mathcal{P}} E_P[\cdot]$
    for some convex set $\mathcal{P}$.
    \begin{itemize}
        \item 
            \textbf{Relative entropy:}
            Let $l(x)=(\exp(\alpha x)-1)/\alpha$ for some $\alpha>0$. Then
            \[\mathcal{OCE}(X)=\sup_{P\in\mathcal{P}} \frac{1}{\alpha}\log E_P[\exp(\alpha X)]   
            =\sup_{Q\in\mathcal{M}_1(\Omega)} \left( E_Q[X] - \inf_{P\in\mathcal{P}} \frac{1}{\alpha}E_Q\Big[ \log \frac{dQ}{dP} \Big]\right).\]
            This is a generalization of the well-known Gibbs variational principle, and 
            $\inf_{P\in\mathcal{P}} E_Q[\log dQ/dP]$
            can be seen as the Kullback-Leibler divergence between the probability measure $Q$ and the set $\mathcal{P}$.
        \item
            \textbf{Monotone mean-variance:}
            Let $l(x) = (((1 + x)^+)^2 - 1)/2$.
            Then
            \[ \mathcal{OCE}(X)=
                \sup_{Q\in\mathcal{M}_1(\Omega)} \left(E_Q[X] -\inf_{P\in\mathcal{P}}E_P\Big[\frac{1}{2}\left(\frac{dQ}{dP}\right)^2-1)\Big]
            \right).\]
            The function $\inf_{P\in\mathcal{P}}E_P[(dQ/dP)^2/2-1)]$ can be seen as the 
            R\`enyi divergence of order $2$ between the probability measure $Q$ and the set $\mathcal{P}$.
        \item
            \textbf{Average value-at-risk:}
            Let $l(x)=x^+/\alpha$ for some $\alpha\in(0,1)$.
            Then
            \[ \mathcal{OCE}(X)
            =\sup\big\{ E_Q[X] : Q\in\mathcal{M}_1(\Omega)\text{ such that } dQ/dP\leq 1/\alpha \text{ for some }P\in\mathcal{P}\big\}.\]
    \end{itemize}
\end{example}

\section{Proofs}
\label{sec:proofs}

\subsection{Proof for Section \ref{sec:main.results}}
\label{sec:proofs.section.wasserstein}

\begin{proof}[of Theorem \ref{thm:integral.neigborhood}]
    For every measurable function $f\colon X\to(-\infty,\infty]$ which is bounded from below, define  
    \[ \Phi(f):=\inf_{\lambda\geq 0} \left( \int f^{\lambda c}\,d\mu_0 + \varphi^\ast(\lambda) \right).\]
    The goal is to apply Choquet's theorem (in the form of Theorem \ref{thm:choquet}) 
    to the functional $\Phi$, which requires to check the following four steps.

    \emph{Step 1: Monotonicity and convexity.}
    If $f\leq g$, then  $f^{\lambda c}\leq g^{\lambda c}$ for any $\lambda$ so that
    $\Phi(f)\leq \Phi(g)$.
    Moreover, for $t\in[0,1]$ and $\lambda',\lambda''\geq0$ it holds that
    \[(f+g)^{\lambda c}\leq tf^{\lambda'c}+(1-t)g^{\lambda''c}
    \quad\text{for }\lambda:=t\lambda'+(1-t)\lambda''\]
    which implies (also using convexity of $\varphi^\ast$) that $\Phi(tf+(1-t)g)\leq t\Phi(f)+(1-t)\Phi(g)$. 
    Finally, for every $m\in\mathbb{R}$, $\lambda\geq0$, and $x\in X$ it holds that
    \[ m^{\lambda c}(x)
        =m-\inf_{y\in X} \lambda c(x,y)
    =m\]
    so that $\Phi(m)=\inf_{\lambda\geq0} (\varphi^\ast(\lambda)-m)=m$. 
    As $\Phi$ is monotone, it follows in particular that $\Phi(f)\in\mathbb{R}$
    whenever $f$ is bounded.

    \emph{Step 2: Continuity from above.}
    Denote by $C_b$ and $U_b$ the set of bounded continuous and upper semicontinuous functions
    from $X$ to $\mathbb{R}$, respectively.
    We show that $\Phi$ is continuous from above on $C_b$.
    Let $\varepsilon>0$ and let $(f_n)$ be a sequence in $C_b$ which decreases pointwise to 0.
    Fix some $m$ such that $f_1\leq m$ and $\lambda>0$ such that $\varphi^\ast(\lambda)<\varepsilon$.
    This is possible because $\varphi$ is not constant by assumption, hence 
    $\varphi^\ast$ is real-valued (and therefore continuous by convexity)
    on some neighborhood of $0$.
    Further fix $k$ such that $\mu_0([-k,k]^c)\leq\varepsilon$.
    By assumption there is $r>0$ such that $\lambda c(x,y)\geq m$ whenever $|x-y|\geq r$, hence
    \[ f_n^{\lambda c}(x)
        =\sup_{y\in[x-r,x+r]} (f_n(y)-\lambda c(x,y))
    \leq \sup_{y\in[x-r,x+r]} f_n(y).  \]
    It follows from Dini's lemma that
    $f_n 1_{[-k-r,k+r]}\leq\varepsilon$ for $n$ large, thus
    \[ f_n^{\lambda c} \leq \varepsilon 1_{[-k,k]} + m 1_{[-k,k]^c} \]
    for $n$ large.
    Therefore
    \begin{align*}
        \Phi(f_n)&\leq \varphi^\ast(\lambda) + \int f_n^{\lambda c}(x)\,\mu_0(dx)
        \leq \varepsilon  + \varepsilon \mu_0([-k,k]) + m\mu_0([-k,k]^c)
        \leq (2+m)\varepsilon
    \end{align*}
    for $n$ large and as $\varepsilon>0$ was arbitrary, $\Phi(f_n)\downarrow 0=\Phi(0)$.

    \emph{Step 3: Continuity from below.}
    We show that $\Phi(f_n)\uparrow \Phi(f)$ whenever $(f_n)$ is a sequence
    of measurable functions $f_n\colon X\to(-\infty,\infty]$ bounded from below
    which increases to $f$.
    Because $\Phi$ is increasing, it suffices to show that $\Phi(f) \le \sup_{n}\Phi(f_n)$.
    Assume that $\sup_n\Phi(f_n)<\infty$, because otherwise there is nothing to prove.
    For every $n$ fix $\lambda_n\geq0$ such that
    \[ \varphi^\ast(\lambda_n) +\int f^{\lambda_n c}_n\,d\mu_0 \leq \Phi(f_n) + \frac{1}{n} \]
    and $m\in\mathbb{R}$ with $m\leq f_1\leq f_n$ so that
    \[ f_n^{\lambda_n c}(x)
    \geq \sup_{y\in X} \left( m-\lambda_n c(x,y)\right)=m.\]
    Note that as $\varphi^\ast$ convex and not constant by assumption,
    $\varphi^\ast(r_n)\to \infty$ for every sequence $(r_n)$ which converges
    to $\infty$. Therefore $(\lambda_n)$ is bounded and, 
    possibly after passing to a subsequence, $(\lambda_n)$ converges to some 
    $\lambda\in[0,\infty)$. 
    Note that  
    \[ f^{\lambda c}(x) = \sup_{y\in X} \lim_n \left( f_n(y)-\lambda_n c(x,y) \right)
        \leq \liminf_n \sup_{y\in X} \left( f_n(y)-\lambda_n c(x,y) \right)
    =\liminf_n f^{\lambda_n c}_n(x)  \]
    for every $x$ and by the same argument 
    $\varphi^\ast(\lambda)\leq\liminf_n \varphi^\ast(\lambda_n)$.
    An application of  Fatou's lemma now implies 
    \[ \Phi(f)
        \leq \varphi^\ast(\lambda) +\int f^{\lambda c}\,d\mu_0
        \leq\liminf_n \left( \varphi^\ast(\lambda_n) +\int f^{\lambda_nc}_n\,d\mu_0\right)
        =\sup_n \Phi(f_n)
    \leq\Phi(f), \]
    where the last inequality holds because $\Phi$ is increasing and $f_n\leq f$ for every $n$.  
    Thus the claim follows.

    \emph{Step 4: Computation of the convex conjugate.}
    We claim that 
    \[ \Phi_{C_b}^\ast(\mu)
        :=\sup_{f\in C_b}\left(  \int f\,d\mu - \Phi(f)\right)
        =\Phi_{U_b}^\ast(\mu):=\sup_{f\in U_b}\left(  \int f\,d\mu - \Phi(f)\right)
    =\varphi( d_c(\mu_0,\mu) )\]
    with the convention that $d_c(\mu_0,\mu)=\infty$ if $\mu$ is not a probability.
    First notice that $0\leq \Phi_{C_b}^\ast\leq\Phi^\ast_{U_b}$ because $\Phi(0)=0$ and $C_b$ is a subset of $U_b$.
    To show that $\Phi^\ast_{U_b}(\mu)\leq \varphi(d_c(\mu_0,\mu))$
    one may assume that $d_c(\mu_0,\mu)<\infty$, otherwise this is trivially satisfied 
    (as by assumption $\varphi(\infty)=\infty$).
    Then there is a probability $\pi$ on $X\times X$ with marginals 
    $\pi(\cdot\times X)=\mu_0$ and $\pi(X\times\cdot)=\mu$
    such that $\int c\,d\pi= d_c(\mu_0,\mu)$, see for instance \cite[Theorem 5.9]{villani2008optimal}.
    For any $f\in U_b$ and $\lambda\ge0$, the pointwise inequality
    \[f(y)\leq \lambda c(x,y)+ f^{\lambda c}(x)\quad\text{for all }x,y\] 
    integrated with respect to $\pi$ yields
    \begin{align}
        \label{eq:weak.dual.Phi}
        \int\nolimits_{X} f(y)\,\mu(dy)
        = \int\nolimits_{X\times X} f(y)\,\pi(dx,dy)
        \leq \lambda d_c(\mu_0,\mu) +\int\nolimits_{X} f^{\lambda c}(x)\,\mu_0(dx).
    \end{align}
    By definition it holds that
    $\lambda d_c(\mu_0,\mu)- \varphi^\ast(\lambda)\leq \varphi(d_c(\mu_0,\mu))$ for all $\lambda\geq0$
    so that $\int f\,d\mu-\Phi(f)\leq \varphi(d_c(\mu_0,\mu))$; 
    hence $\Phi^\ast_{C_b}(\mu)\leq \Phi^\ast_{U_b}(\mu)\leq \varphi(d_c(\mu_0,\mu))$.

    To show the reverse inequality, fix some $\mu$ and $\varepsilon>0$.
    By the Fenchel-Moreau theorem, $\varphi(x)=\sup_{\lambda\geq 0}(\lambda x-\varphi^\ast(x))$
    for every $x\geq 0$, hence there is $\lambda\geq 0$ such that 
    \begin{align*}
        \lambda d_c(\mu_0,\mu)-\varphi^\ast(\lambda)
        &\geq \varphi(d_c(\mu_0,\mu))-\varepsilon, 
        & &\quad\text{if } \varphi(d_c(\mu_0,\mu))<\infty,\\
        \lambda d_c(\mu_0,\mu)-\varphi^\ast(\lambda)
        &\geq 1/\varepsilon,
        & &\quad\text{if } d_c(\mu_0,\mu)<\infty \text{ but } \varphi(d_c(\mu_0,\mu))=\infty,\\
        \varphi^\ast(\lambda)
        &\leq\varepsilon \text{ and }\lambda>0
        & &\quad\text{if } d_c(\mu_0,\mu)=\infty.
    \end{align*}
    As $d_{\lambda c}(\mu,\mu_0)=\lambda d_c(\mu_0,\mu)$,
    by the dual formula \eqref{eq:dual.for.transport} for $d_c$, 
    there are bounded and continuous functions $f,g$ such that
    \[f(x)+g(y)\leq \lambda c(x,y) \quad \text{for all } x,y
        \text{ and}\quad
        \int f\,d\mu +\int g\,d\mu_0
        \geq 
        \begin{cases}
            \lambda d_c(\mu_0,\mu)-\varepsilon,&\text{if } d_c(\mu_0,\mu)<\infty,\\
            1/\varepsilon,&\text{otherwise}.
    \end{cases}\]
    Note that $f(x)-\lambda c(x,y)\leq -g(y)$ from which it follows that
    $f^{\lambda c}\leq -g$. As $\varphi(d_c(\mu_0,\mu))<\infty$ implies
    $d_c(\mu_0,\mu))<\infty$ we therefore get
    \[ \int f\,d\mu-\Phi(f)
        \geq \int f\,d\mu +\int g\,d\mu_0 -\varphi^\ast(\lambda)
        \geq
        \begin{cases}
            \varphi(d_c(\mu_0,\mu))-2\varepsilon,&\text{if } \varphi(d_c(\mu_0,\mu))<\infty,\\
            1/\varepsilon-\varepsilon,  &\text{if } \varphi(d_c(\mu_0,\mu))=\infty.
    \end{cases}\]
    As $\varepsilon>0$ was arbitrary, it follows that $\Phi^\ast_{C_b}(\mu)\geq \varphi(d_c(\mu_0,\mu))$
    which by the previous part shows 
    $\Phi^\ast_{C_b}(\mu)=\Phi^\ast_{U_b}(\mu)=\varphi(d_c(\mu_0,\mu))$.
    The representation \eqref{eq:rep.ball} now follows from an application of Theorem \ref{thm:choquet}.

    As for the existence of an optimal $\lambda\geq0$, apply \emph{Step 3} to the constant sequence $f_n=f$.
\end{proof}

\begin{proof}[of Theorem \ref{thm:oce.and.es.robust.compute}]   
    First note that as $c$ depends only on the difference, one has
    $l(\cdot-m)^{\lambda c}(x)=l^{\lambda c}(x-m)$
    for all $m\in \mathbb{R}$, $x\in \mathbb{R}$, and $\lambda\ge 0$.
    Now, by Theorem \ref{thm:integral.neigborhood} one has
    \begin{align*}
    \mathcal{OCE}(l)
    &=\inf_{m\in\mathbb{R}} \Big( \sup_{\mu\in\mathcal{M}_1(\mathbb{R})} \Big(
    \int l(x-m)\,\mu(dx) -\varphi\big( d_c(\mu_0,\mu))\big)\Big)+m \Big)\\
    &=\inf_{m\in\mathbb{R}}\Big(\inf_{\lambda\geq 0}\Big(\int l(\cdot-m)^{\lambda c}(x) \,\mu_0(dx)+ \varphi^\ast(\lambda)\Big) +m\Big)\\
    &=\inf_{\lambda\geq 0} \inf_{m\in\mathbb{R}} \Big(\int l^{\lambda c}(x-m) \,\mu_0(dx) +\varphi^\ast(\lambda)+m\Big)
    =\inf_{\lambda\geq 0} \Big(\mathrm{OCE}(l^\lambda)+\varphi^\ast(\lambda)\Big)
    ,
    \end{align*}
    which completes the proof.
    The same arguments show that
    \begin{align*}
    \mathcal{ES}(l)
    &=\inf\Big\{ m\in\mathbb{R} : \min_{\lambda\geq0} \Big( 
    \int l^{\lambda c}(x-m)\,\mu_0(dx) + \varphi^\ast(\lambda)\Big) \leq 0\Big\}\\
    &=\inf_{\lambda\geq0} \inf\Big\{ m\in\mathbb{R} : 
    \int l^{\lambda c}(x-m)\,\mu_0(dx)+ \varphi^\ast(\lambda)\leq 0\Big\}\\
    &=\inf_{\lambda\geq0} \mathrm{ES}(l^{\lambda c}+\varphi^\ast(\lambda)).
    \end{align*}
\end{proof}

\begin{proof}[of Remark \ref{rem:c.bounded.or.different.metrik}]
    We first prove that $\mathcal{OCE}(l)=\mathcal{ES}(l)=\infty$ if 
    $\liminf_{x\to\infty} c(x)<\infty$ and $l$ is a loss function.
    Because $l$ is increasing, convex and not constant, there exist
    $a,b>0$ such that $l(x)\geq a x-b$ for every $x\in\mathbb{R}$.
    Because $\varphi$ is continuous at 0, there is $\delta>0$ such that
    $\varphi(\delta)<\infty$.
    Moreover, by assumption, there is some $r\in\mathbb{R}$ and a sequence 
    $(x_k)$ in $\mathbb{R}$ such that $x_k\geq k$ and $c(x_k)\leq r$.
    For simplicity let us assume that $c(0)=0$, $\mu_0=\delta_0$, and define
    $\mu_k=(1-\delta/r)\delta_0+\delta/r\delta_{x_k}$. Then
    \[d_c(\mu_k,\mu_0)= \left(1-\frac{\delta}{r}\right) c(0-0)+\frac{\delta}{r} \tilde{c}(x_k-0)
    \leq \delta\]
    so that $\varphi(d_c(\mu_k,\mu_0))\leq\varphi(\delta)<\infty$.
    However, as 
    \[\sup_k \int l(x-m)\,\mu_k(dx)\\
        \geq \sup_k \left( \left(1-\frac{\delta}{r}\right)l(0-m) + \frac{\delta}{r} \left(a (x_k-m)-b\right)\right)
    =\infty\]
    for every $m \in \mathbb{R}$, it follows that $\mathcal{OCE}(l)=\mathcal{ES}(l)=\infty$.

    To show that $\mathcal{OCE}(l)=\mathcal{ES}(l)=\infty$ if $d_c$ is replaced by a distance
    compatible with weak convergence, let $\mu_k=(k-1)/k\delta_0+1/k\delta_{k^2}$
    which converges weakly to $\delta_0$.
    As $\varphi$ is continuous at 0, $\varphi(d(\mu_k,\delta_0))\to 0$.
    However, $\sup_k  \int l(x-m)\,\mu_k(dx)=\infty$ for every $m$, 
    which again  implies $\mathcal{OCE}(l)=\mathcal{ES}(l)=\infty$.
\end{proof}

\begin{proof}[of Example \ref{ex:avar.formula}]
    For every $\lambda\geq 0$ it holds that
    \[ \sup_{y\in\mathbb{R}} \left( \frac{1}{\alpha}y^+-\lambda(x-y)^2\right)
    = \frac{1}{\alpha}\left( x+\frac{1}{4\lambda\alpha}\right)^+\]
    so that
    \[\mathrm{OCE}(l^{\lambda c})
    =\mathrm{OCE}(l) + \frac{1}{4\lambda\alpha} \]
    for every $\lambda\geq 0$.
    Thus, Theorem \ref{thm:oce.and.es.robust.compute} yields 
    \[ \mathcal{OCE}(l)
    =\mathrm{OCE}(l) + \inf_{\lambda\geq 0} \left( \varphi^\ast(\lambda)  + \frac{1}{4\lambda\alpha}\right).\]
    It remains to plug the different $\varphi$'s in and compute the infimum.
    This proves the claim for $p=2$.

    Similarly, for $p=1$, it holds that $l^{\lambda c}(x) = \infty$ if $\lambda < 1/\alpha$ and $l^{\lambda c}(x) = x^+/\alpha$ else.
    Thus, it follows by Theorem \ref{thm:oce.and.es.robust.compute} that
    \[\mathcal{OCE}(l) 
        = \mathrm{OCE}(l) + \inf_{\lambda \geq 1/\alpha}\varphi^\ast(\lambda)
    =\mathrm{OCE}(l) + \varphi^\ast\left(\frac{1}{\alpha}\right)\]
    where the last equality holds because $\varphi^\ast$ is increasing.
\end{proof}

\begin{proof}[of Example \ref{ex:mean-var}]
    It holds
    \begin{equation*}
        l^{\lambda c}(x)=
        \begin{cases}
            \infty,& \text{if } \lambda <1/2,\\
            \frac{2\lambda}{2\lambda - 1}l(x)+ \frac{1}{4\lambda - 2}, & \text{else}.
        \end{cases}
    \end{equation*}
    Thus, for the optimized certainty equivalent, one has 
    $\mathrm{OCE}(l^{\lambda c})=\mathrm{OCE}(\frac{2\lambda}{2\lambda - 1}l)+\frac{1}{4\lambda - 2}$ 
    so that by Theorem \ref{thm:oce.and.es.robust.compute} it holds that
    \begin{equation*}
        \mathcal{OCE}(l)
        =\inf_{\lambda>1/2}\left( \mathrm{OCE}\left(\frac{2\lambda}{2\lambda-1} l \right) 
        +\frac{1}{4\lambda-2} - \varphi^\ast(\lambda) \right).
    \end{equation*}
\end{proof}

\begin{proof}[of Example \ref{ex:var}]
    Note that the value at risk is a special case of the expected shortfall, corresponding
    to the loss function $l=1_{(0,\infty)}-\alpha$.
    Further, with the convention $0^{-1/p}=\infty$, it holds that
    \[l^{\lambda c}(x)
    =l(x) + (1-\lambda|x|^p)1_{(-\lambda^{-1/p},0]}(x)-\alpha\]
    for every $x$.
    Therefore $\int l^{\lambda c}(x-m)\,\mu_0(dx)=\mu_0((m,\infty))+e(m,\lambda)-\alpha$
    so that
    \begin{align*}
        \mathcal{V@R}_\alpha
        &=\mathcal{ES}(l)
        =\inf_{\lambda\geq0} \mathrm{ES}\left(l^{\lambda c}+\varphi^\ast(\lambda)\right)\\
        &=\inf_{\lambda\geq0}\inf\big\{ m\in\mathbb{R} : \mu_0((m,\infty))+ e(m,\lambda)+\varphi^\ast(\lambda)\leq\alpha\big\},
    \end{align*}
    by Theorem \ref{thm:oce.and.es.robust.compute}.
    The special case $\varphi(x)=x$ follows from Remark \ref{rem:special.case.varphi.oce}.
\end{proof}

\begin{proof}[of Proposition \ref{prop:pricing.dual}]
    The proof is similar to the one of Theorem \ref{thm:integral.neigborhood}
    and we only give a sketch. 
    Denote by $B_{\mathrm{lin}}$ the set of measurable functions $f\colon \mathbb{R}^d\to\mathbb{R}$
    for which $\sup_x f(x)/(1+|x|)$ is finite, and by
    $U_{\mathrm{lin}}$ and $C_{\mathrm{lin}}$ the subsets of upper semicontinuous and continuous functions,
    respectively.
    For every $f\in B_{\mathrm{lin}}$ define
    \[ \Phi(f):=\inf_{\alpha\in\mathbb{R}^d,\lambda\geq 0}
    \left( \int f^{\lambda c, \alpha}\,d\mu_0 -\varphi\left( d_c(\mu_0,\mu))\right) \right) \]
    which is well-defined as 
    $f^{\lambda c,\alpha}(x)
    \geq f(x)-\alpha\cdot (x-s) -c(0)
    \geq -k(|x|+1)$ for some $k>0$ and $\int |x|\,\mu_0(dx)<\infty$.

    As in Theorem \ref{thm:integral.neigborhood} one checks that $\Phi$ is a 
    convex function on $B_{\mathrm{lin}}$ which is continuous from above on $C_{\mathrm{lin}}$
    (here the growth condition $\liminf_{|x|\to\infty} c(x)/(1+|x|^{1+\varepsilon})=\infty$
    is used).
    Moreover, similar arguments as in Theorem \ref{thm:integral.neigborhood}
    show that $\Phi$ is continuous from below on $B_{\mathrm{lin}}$
    (here the condition that $\int c(x-y)\,\mu_0(dx)<\infty$ for every $y$ is used).
    By a version of Theorem \ref{thm:choquet} (see \cite[Theorem 2.2]{bartl2017robust}), 
    it follows that
    \[ \Phi(f)=\sup_{\mu\in\mathcal{M}_1(\mathbb{R}^d)}\left( \int f\,d\mu -
    \Phi^\ast_{C_{\mathrm{lin}}}(\mu) \right)\quad\text{for } f\in B_{\mathrm{lin}}, \]
    provided that
    \[\Phi^\ast_{C_{\mathrm{lin}} }(\mu)
        :=\sup_{f\in C_{\mathrm{lin}} } \left(\int f\,d\mu -\Phi(f) \right)
        =\sup_{f\in U_{\mathrm{lin}} } \left(\int f\,d\mu -\Phi(f) \right)
    =:\Phi^\ast_{U_{\mathrm{lin}} }(\mu).  \]
    As $\Phi(\alpha\cdot (S-s))\leq 0$ and $\alpha\cdot (S-s)\in C_{\mathrm{lin}}$
    for every $\alpha\in\mathbb{R}^d$, similar computations as in 
    Theorem \ref{thm:integral.neigborhood} yield
    \[\Phi^\ast_{C_{\mathrm{lin}} }(\mu)
        =\Phi^\ast_{U_{\mathrm{lin}} }(\mu)
        =\begin{cases}
            \varphi(d_c(\mu_0,\mu))), &\text{if }\int S-s\,d\mu=0\\
            \infty, &\text{otherwise}.
    \end{cases}\]
    This ends the proof.
\end{proof}

\begin{proof}[of Example \ref{ex:call.price}]
    We compute 
    \[h^{\lambda c,\alpha}(x)=\sup_{y\in\mathbb{R}}\left( (y-k)^+ +\alpha (y-s)-\frac{\lambda}{2}( y-x )^2\right).\]
    The first order conditions yield
    \[1_{(0,\infty)}(y-k)+\alpha-\lambda(y-x)
        \geq 0
    \geq 1_{[0,\infty)}(y-k)+\alpha-\lambda (y-x)\]
    For an optimizer there are three cases:
    \begin{itemize}
        \item 
            $y^\ast <k$, then $y^\ast-x=\alpha/\lambda$, hence
            $y^\ast=\alpha/\lambda +x$ with $x \in (-\infty, k-\alpha/\lambda)$.
        \item $y^\ast>k$, then $y^\ast-x=(\alpha+1)/\lambda$, hence $y^\ast=(\alpha+1)/\lambda+x$ with $x\in (k-(\alpha+1)/\lambda, \infty)$.
        \item $y^\ast=k$ is impossible.
    \end{itemize}
    Hence we have three cases:
    \begin{itemize}
        \item 
            For $x\in(-\infty, k-(\alpha+1)/\lambda]\subseteq (-\infty, k-\alpha/\lambda)$, 
            it follows that
            \[h^{\lambda c,\alpha}(x)
                =\alpha\left(\frac{\alpha}{\lambda}+x-s\right)-\frac{\lambda}{2}\frac{\alpha^2}{\lambda^2}
            =\alpha(x-s)+\frac{\alpha^2}{2 \lambda}.\]
        \item 
            For $x\in[k-\alpha/\lambda, \infty)\subseteq (k-(\alpha+1)/\lambda, \infty)$, it follows that
            \begin{align*}
                h^{\lambda c,\alpha}(x) 
                & = \left( \frac{\alpha+1}{\lambda}+x-k \right)+\alpha\left( \frac{\alpha+1}{\lambda} +x-s\right)-\frac{(\alpha+1)^2}{2\lambda}\\
                & = \left( x-\left( k-\frac{2\alpha+1}{2\lambda} \right) \right)+\alpha(x-s)+\frac{\alpha^2}{2\lambda}.
            \end{align*}
        \item 
            For $x\in(k-(\alpha+1)/\lambda,k-\alpha/\lambda)$, it follows that
            \[h^{\lambda c,\alpha}(x)=\left( \alpha(x-s)+\frac{\alpha^2}{2 \lambda} \right)\vee \left(\left( x-\left( k-\frac{2\alpha+1}{2\lambda} \right) \right)+\alpha(x-s)+\frac{\alpha^2}{2\lambda}  \right).\]
    \end{itemize}
    Hence
    \[ h^{\lambda c,\alpha}(x)=\left( x-\left( k-\frac{2 \alpha+1}{2 \lambda} \right) \right)^+
    +\alpha (x-s)+\frac{\alpha^2}{2\lambda}.\]
    Therefore, Proposition \ref{prop:pricing.dual} and the fact that $\mu_0$ is a pricing measure yield
    \[ \mathcal{CALL}(k)=\inf_{\alpha\in\mathbb{R}}\inf_{\lambda> 0}\left(
            \mathrm{CALL}\left( k-\frac{2\alpha +1}{2\lambda}\right)
    +\frac{\alpha^2}{2\lambda}+\varphi^\ast(\lambda) \right).\]
    Plugging the special cases of $\varphi$ into this equation yields the claim.
\end{proof}

\subsection{Proofs for Section \ref{sec:directed}}

The main argument for the proof of Proposition \ref{prop:avar-int} is given in the next lemma.
\begin{lemma}
    \label{lem:sup_integr}
    Assume that $\mathcal{D}$ satisfies \eqref{eq:ass.dir}.
    Then, there exists $\mu^*\in\mathcal{M}_1$ such that
    \begin{align}
        \label{eq:int_ubarF}
        \sup_{\mu \in {\cal D}}\int f d\mu=\int fd\mu^*
    \end{align}
    for every increasing, continuous function $f\colon\mathbb{R}\to \mathbb{R}$ that is bounded from below.
    If in addition $\mathcal{D}$ is closed in the weak topology induced by all continuous bounded functions, then 
    $\mu^\ast\in\mathcal{D}$.
\end{lemma}
\begin{proof}
    First assume that $f$ is bounded.
    As $\mathcal{D}$ is tight, it can be checked that $\ubar{F}$ defined by
    $\ubar{F}(t):=\inf_{\mu \in {\cal D}}F_\mu(t)$ where $F_\mu(t):=\mu(-\infty, t]$ is a cumulative distribution function.
    Furthermore, $f$ being increasing, continuous and bounded,
    it defines a finite Borel measure $df$ on the real line.
    Hence $df$ is regular and $\tau$-additive, see for instance \cite[Proposition 7.2.2]{bogachev2}.
    Let us first show that 
    \begin{align}
        \label{eq:infD-int}
        \int\ubar{F}df = \inf_{\mu \in \mathcal{D}}\int F_\mu df.
    \end{align}
    Each cumulative distribution function $F_\mu$ is increasing and right-continuous, hence upper semicontinuous.
    Because ${\cal D}$ satisfies \eqref{eq:ass.dir}, the net\footnote{$\mathcal{D}$ is endowed with the ordering $\mu\preceq\nu$ if and only if $\mu(-\infty,t]\geq \nu(-\infty,t]$ for every $t$.}
    $(F_\mu)_{\mu \in {\cal D}}$ is decreasing.
    Thus, $(1-F_\mu)_{\mu}$ is an increasing net of nonnegative lower semicontinuous functions such that 
    $1-\ubar{F} = \lim_{\mu} (1 - F_\mu)$.
    It therefore follows from  \cite[Lemma 7.2.6]{bogachev2} that
    \begin{align*}
        \sup_{\mu \in {\cal D}}\int 1 - F_\mu \,df = \lim_{\mu} \int 1 - F_\mu\, df = \int 1 - \ubar{F}\,df,
    \end{align*}
    which shows \eqref{eq:infD-int}.
    Moreover, as $f$ is continuous, one has $\int \ubar{F}(x-)\,df(x)= \int \ubar{F}(x)\,df(x)$.
    Hence, integration by parts yield
    \begin{align*}
        \int fd\ubar{F}
        &=f(\infty)-\int \ubar{F}df
        =f(\infty)-\int\inf_{\mu \in \mathcal{D}} F_\mu \,df
        =f(\infty)-\inf_{\mu \in \mathcal{D}} \int F_\mu df\\
        &=f(\infty)-\inf_{\mu \in \mathcal{D}} \left( f(\infty) -\int fdF_{\mu} \right)
        =\sup_{\mu\in \mathcal{D}}\int fdF_{\mu},
    \end{align*}
    showing \eqref{eq:int_ubarF} whenever $f$ is bounded, with $\mu^*$ being the distribution associated to $\ubar{F}$.
    If $f$ is not bounded, we approximate $f$ from below by  $f^n:=f\wedge n$.

    If $\mathcal{D}$ is also closed, then it follows from Prokhorov's theorem
    and tightness that $\mathcal{D}$ is compact.
    Suppose for contradiction that $\mu^\ast\notin\mathcal{D}$.
    Then, by the strong separation theorem and \eqref{eq:int_ubarF} which was already proven, 
    there exists a continuous bounded and increasing function $f\colon\mathbb{R}\to\mathbb{R}$ such that 
    $\int f d\mu^\ast > \sup_{\mu \in \mathcal{D}} \int f d\mu$, which clearly contradicts \eqref{eq:int_ubarF}.
    Thus, $\mu^\ast\in\mathcal{D}$.
\end{proof}

\begin{corollary}
    \label{cor:when_D_directed}
    Assume that \eqref{eq:ass.dir} holds and that $l\colon\mathbb{R}\to\mathbb{R}$ is convex, increasing, bounded from below, 
    and that $l(x) > x$ for $|x|$ large enough.
    Then there exists $\mu^*\in\mathcal{M}_1$ and $m^*\in \mathbb{R}$ such that 
    \[\mathcal{OCE}(l)
        =\inf_{m\in\mathbb{R}} \left( \int l(x-m)\,\mu^\ast(dx)+m\right)
    =\int l(x-m^\ast)\,\mu^\ast(dx) +m^\ast. \]
    In particular $m^\ast$ is characterized by 
    \[ \int l'_-(x-m^\ast)\,\mu^\ast(dx) 
        \leq 1
    \leq \int l'_+(x-m^\ast)\,\mu^\ast(dx)\]
    for the right and left hand derivatives $l'_-$ and $l'_+$ of $l$.
    If $l$ is continuously differentiable, then inequalities in the above formula are equalities.
\end{corollary}

\begin{proof}
    The existence of a $\mu^\ast\in\mathcal{M}_1$ such that
    $\mathcal{OCE}(l)
    =\inf_{m\in\mathbb{R}} ( \int l(x-m)\,\mu^\ast(dx)+m)$
    follows directly from Lemma \ref{lem:sup_integr}.
    Therefore, the existence and characterization of an optimal allocation $m^*$ can be deduced from the non-robust case,
    see for instance \cite{Ben-Teb}.
\end{proof}

\begin{proof}[of Proposition \ref{prop:avar-int}]
    It follows from Lemma \ref{lem:sup_integr} and Corollary \ref{cor:when_D_directed} that one has
    $\mu^*(-\infty,t] = \inf_{\mu \in {\cal D}}\mu(-\infty,t]$ for every $t$ which implies 
    $\mathcal{V@R}_\alpha=\inf\{m \in\mathbb{R}: \mu^\ast((m,\infty))\leq\alpha\}$ and 
    $\mathcal{AV@R}_\alpha= \inf_{m\in\mathbb{R}}(\int (x-m)^+/\alpha\,\mu^\ast(dx)+m)$.
    Thus, \cite[Lemma 4.51]{foellmer11} yields
    \begin{equation*}
        \mathcal{AV@R}_\alpha
        =\frac{1}{\alpha} \int_0^\alpha \inf\{ m\in\mathbb{R} :\mu^\ast((m,\infty))\leq u\} \,du
        =\frac{1}{\alpha}\int_0^\alpha \mathcal{V@R}_u\,du.
    \end{equation*}
\end{proof}

\begin{proof}[of Example \ref{exa:dir}]
    For every $\ubar{b}\le b\le \bar{b}$, the process
    $S_t^b = S_0\exp((b - \frac{1}{2}\sigma^2)t+ \sigma W_t)$ is the solution of
    $dS_t = S_t(b\,dt + \sigma d W_t)$.
    Further, as $S_0>0$ and $f$ is strictly increasing, one has
    \begin{align}
        \label{eq:drift.uncertainty}
        P( f(S_T^{\bar{b}})\leq x)
        \leq P(f(S_T^b)\leq x)
        \leq P(f(S_T^{\ubar{b}})\le x)
    \end{align}
    for all $x$.
    Thus a straightforward computation shows that $\mathcal{D}$ satisfies \eqref{eq:ass.dir}.
\end{proof}
\begin{proof}[of Corollary \ref{cor:kusuoka}]
    By Lemma \ref{lem:sup_integr}, there is $\mu^*$ such that ${\cal AV@R}_u = \mathrm{AV@R}_u(\mu^*)$ for every $u \in (0,1]$.
    Hence, it follows by definition that 
    \begin{equation*}
        \rho({\cal D})\leq \sup_{\nu \in {\cal M}_1((0,1])}\left(\int\nolimits_{(0,1]}\mathcal{AV@R}_u\,\nu(du) - \beta(\nu) \right).
    \end{equation*}
    On the other hand, as ${\cal D}$ is closed, it holds that $\mu^* \in {\cal D}$ and thus $\sup_{\mu \in {\cal D}}\int_{(0,1]}AV@R_u(\mu)\,\nu(du)\ge \int_{(0,1]}AV@R_u(\mu^*)\,\nu(du)= \int_{(0,1]}{\cal AV@R}_u\,\nu(du)$, which proves the reverse inequality.
\end{proof}

\begin{proof}[of Theorem \ref{thm:dual}]
    Denote by 
    $\mathrm{dom}(\beta):=\{P\in\mathcal{M}_1(\Omega) : \beta(P)<\infty\}$ the domain of $\beta$
    which is a nonempty convex set by assumption. 
    Therefore the function $J$ defined by
    \[  J\colon\mathbb{R}\times\mathrm{dom}(\beta)\to\mathbb{R},
    \quad(m,P)\mapsto E_P[l(X-m)]+m - \beta(P) \]
    is convex and continuous in $m$, and concave in $P$.
    Moreover, as $X$ is a bounded random variable and 
    $l$ is bounded from below with $\lim_{x\to\infty} l(x)/x=\infty$
    by assumption, there exists some $m_0\in\mathbb{R}$ such that
    the first and last equality in the following equation hold
    \begin{align*} 
        \inf_{m\in\mathbb{R}}\sup_{P\in\mathrm{dom}(\beta)}  J(m,P)
        &=\inf_{m\in[-m_0,m_0]}\sup_{P\in\mathrm{dom}(\beta)}  J(m,P)\\
        &=\sup_{P\in\mathrm{dom}(\beta)}\inf_{m\in[-m_0,m_0]}  J(m,P)
        =\sup_{P\in\mathrm{dom}(\beta)}\inf_{m\in\mathbb{R}}  J(m,P).
    \end{align*}
    The middle equality follows from a minimax theorem, see for instance \cite[Theorem 2]{fan53}.
    Now notice that the left hand side equals $\mathcal{OCE}(X)$
    and the right hand side the supremum over $P\in\mathrm{dom}(\beta)$
    of the optimized certainty equivalent under $P$.
    In particular, it follows from the classical representation of the optimized certainty
    equivalent that
    \begin{align*} 
        \mathcal{OCE}(X)
        &=\sup_{P\in\mathrm{dom}(\beta)} \sup_{Q\in\mathcal{M}_1(\Omega)}\left( E_Q[X]-E_P\Big[ l^\ast\left(\frac{dQ}{dP}\right) \Big] - \beta(P)\right)\\
        &=\sup_{Q\in\mathcal{M}_1(\Omega)}\left( E_Q[X]-\inf_{P\in\mathcal{M}_1(\Omega)}\left(E_P\Big[ l^\ast\left(\frac{dQ}{dP}\right)\Big]+\beta(P)\right)\right).
    \end{align*}
\end{proof}

\begin{appendix}
    \section{Appendix}

    Let $X\subset\mathbb{R}^d$ be closed, denote by $B_{b-}$ the set of all measurable functions from $X$ to $(-\infty,\infty]$ which are bounded from below, and by $U_b$ and $C_b$ the subsets of bounded upper semicontinuous (resp. bounded continuous) functions.
    Further write $\mathcal{M}(X)$ for the set of all countably additive, finite, positive Borel measures on $X$,
    that is $\mathcal{M}(X):=\{t\mu :\mu\in\mathcal{M}_1(X) \text{ and } t\geq0\}$.
    The following theorem, which builds on Choquet's theory on the regularity of capacities, 
    is a slight modification of Theorem 2.2 in \cite{bartl2017robust}. 

    \begin{theorem}
        \label{thm:choquet} 
        Let $\Phi\colon B_{b-}\to(-\infty,\infty]$ be a monotone convex functional such that $\Phi(f)<\infty$
        whenever $f$ is bounded.
        If 
        \begin{itemize}
            \item
                $\Phi(f_n)\downarrow \Phi(0)$ for every sequence $(f_n)$ in $C_b$ which decreases pointwise to 0,
            \item
                $\Phi^\ast_{C_b}(\mu):=\sup_{f\in C_b} (\int f\,d\mu - \Phi(f))=
                \sup_{f\in U_b} (\int f\,d\mu - \Phi(f))=:\Phi^\ast_{U_b}(\mu)$
                for every $\mu\in \mathcal{M}(X)$,     
            \item
                $\Phi(f_n)\uparrow \Phi(f)$ for every sequence $(f_n)$ in $B_{b-}$ which increases pointwise to $f\in B_{b-}$,
        \end{itemize}
        then
        \begin{align}
            \label{eq:choquet.dual.rep}
            \Phi(f)=\sup_{\mu\in\mathcal{M}(X)} \left( \int f \,d\mu -\Phi^\ast_{C_b}(\mu)\right)     \quad \text{for every } f\in B_{b-}.
        \end{align}
    \end{theorem}
    \begin{proof}
        If $B_{b-}$ is replaced by the set of all bounded measurable functions this is exactly
        the statement of \cite[Theorem 2.2]{bartl2017robust},
        so that \eqref{eq:choquet.dual.rep} holds for all bounded measurable $f$.
        For general $f\in B_{b-}$, notice that \eqref{eq:choquet.dual.rep} holds for every 
        $f\wedge n$, hence the claim follows from the third assumption 
        $\Phi(f)=\sup_n\Phi(f\wedge n)$, interchanging two suprema, and the monotone
        convergence theorem applied under each $\mu\in\mathcal{M}(X)$.
    \end{proof}

\end{appendix}

\bibliographystyle{abbrvnat}

\end{document}